\documentclass[a4paper,english,numberwithinsect,cleveref,thm-restate]{lipics-v2021}
\usepackage{inputenc}[utf8]

\usepackage{microtype} 
\usepackage{xspace}
\usepackage{csquotes}
\usepackage[noabbrev,capitalise]{cleveref}
\usepackage[dvipsnames]{xcolor}

\graphicspath{{figures/}}

\usepackage[titletoc]{appendix}

\usepackage {amssymb}
\usepackage {amsmath}
\usepackage {enumerate}
\usepackage{graphicx}  
\usepackage{hyperref}
\usepackage{algorithm}
\usepackage{algorithmic}
\usepackage{subcaption}
\usepackage{xspace}
\usepackage[capitalise]{cleveref}
\usepackage{todonotes}
\usepackage{csquotes}

\nolinenumbers
\hideLIPIcs

\newcommand{\R}{\ensuremath{\mathbb{R}}\xspace}

\newcommand{\mkmcal}[1]{\ensuremath{\mathcal{#1}}\xspace}
\newcommand{\D}{\mkmcal{D}}
\newcommand{\NN}{N\hspace{-7.5pt}N}

\DeclareMathOperator{\polylog}{polylog}

\title{Nearest-Neighbor Decompositions of Drawings\footnote{This research 
was started at the 4th DACH Workshop on
    Arrangements and Drawings, February 24--28, 2020, in Malchow,
    Germany.  We thank 
    all participants of the workshops for valuable discussions
    and for creating a conducive research atmosphere.}
}
\titlerunning{Nearest-Neighbor Decompositions of Drawings}

\author{Jonas Cleve}
{Institut f\"ur Informatik, Freie Universität Berlin}
{jonascleve@inf.fu-berlin.de}
{https://orcid.org/0000-0001-8480-1726}
{Supported in part by ERC StG 757609.}
\author{Nicolas Grelier}
{ETH Zürich, Department of Computer Science}
{nicolas.grelier@inf.ethz.ch}
{}
{Supported by the Swiss National Science Foundation 
within the collaborative DACH project 
\emph{Arrangements and Drawings} as SNSF Project 200021E-171681.}
\author{Kristin Knorr}
{Institut f\"ur Informatik, Freie Universität Berlin}
{knorrkri@inf.fu-berlin.de}
{https://orcid.org/0000-0003-4239-424X}
{Supported
by the German Science Foundation within the research
training group `Facets of Complexity' (GRK 2434).}
\author{Maarten L\"offler}
{Utrecht University}
{m.loffler@uu.nl}
{}
{}
\author{Wolfgang Mulzer}
{Institut f\"ur Informatik, Freie Universität Berlin}
{mulzer@inf.fu-berlin.de}
{https://orcid.org/0000-0002-1948-5840}
{Supported in part by ERC StG 757609
and by the German Research Foundation within the collaborative 
DACH project \emph{Arrangements
      and Drawings} as DFG Project MU 3501/3-1.}
\author{Daniel Perz}
{Graz University of Technology}
{daperz@ist.tugraz.at}
{https://orcid.org/0000-0002-6557-2355}
{Partially supported by FWF within the 
collaborative DACH project \emph{Arrangements and Drawings} as 
FWF project \mbox{I 3340-N35}}

\authorrunning{J. Cleve, N. Grelier, K. Knorr, M. L\"offler,
W. Mulzer, and D. Perz}
\Copyright{Jonas Cleve, Nicolas Grelier, Kristin Knorr, Maarten L\"offler,
Wolfgang Mulzer, and Daniel Perz}



\begin{document}
\maketitle
\begin{abstract}
Let $\D$ be a set of straight-line segments in the plane, potentially crossing, and let
$c$ be a positive integer. We denote by $P$ the union of the endpoints of the straight-line segments of $\D$ and of the intersection points between pairs of segments.
We say that $\D$ has a \emph{nearest-neighbor decomposition} into
$c$ parts if we can partition $P$ into $c$ point sets $P_1, \dots, P_c$ such that
$\D$ is the union of the nearest neighbor graphs
on $P_1, \dots, P_c$.
We show that it is NP-complete to decide whether $\D$  can be drawn
as the union of $c\geq 3$ nearest-neighbor graphs, even when no two segments cross. We show that for $c = 2$, it is NP-complete in the general 
setting and polynomial-time solvable when no two segments cross. We show the existence of an $O(\log n)$-approximation algorithm running in subexponential time for partitioning $\D$ into a minimum number of nearest-neighbor graphs.

As a main tool in our analysis, we establish the notion of the 
\emph{conflict graph} for a drawing $\D$.
The vertices of the conflict graph are the connected components
of $\D$, with the assumption that each connected component is the nearest neighbor graph of its vertices, and there is an edge between two components $U$ and
$V$  if and only if the nearest neighbor graph of $U \cup V$
contains an edge between a vertex in $U$ and a vertex in $V$.
We show that string graphs are conflict graphs of certain planar drawings.
For planar graphs and complete $k$-partite graphs, we give additional, more efficient constructions.
We furthermore show that there are subdivisions of 
non-planar graphs that are not conflict graphs. Lastly, we show a separator lemma for conflict graphs.
\end{abstract}

\section{Introduction}
\label{sec:Introduction}

Let $P \subset \R^2$ be a finite planar point set, and let
$C$ be a finite set of \emph{colors}. 
A \emph{coloring}
is a function $\sigma: P \to C$ that assigns a color
to each point in $P$.
For any color
$c \in C$, we write $P_c = \{p \in P \mid \sigma(p) = c\}$ for
the points in $P$ that were colored with $c$.

In the following, we assume all pairwise distances
in $P$ are distinct.
The \emph{nearest-neighbor graph for a color $c \in C$}, $\NN_c$,
is the embedded graph with vertex set $P_c$ and a straight-line
edge between $p, q \in P_c$ if and only if $p$ is the
nearest neighbor of $q$ among all points in $P_c$, or vice
versa.\footnote{Our notion of nearest-neighbor graph is undirected, but
a directed version also exists.} We will consider
$\NN_c$ both as a combinatorial graph, consisting of vertices 
and edges, and as a
subset of the plane, consisting of the points in $P_c$ and the line
segments that represent the edges.
We write $\NN = \bigcup_{c \in C} \NN_c$ for the union of
the nearest-neighbor graphs of all colors. Again, we consider
$\NN$ both as a graph and as a set.

\begin{figure}[h]
  \center
  \includegraphics{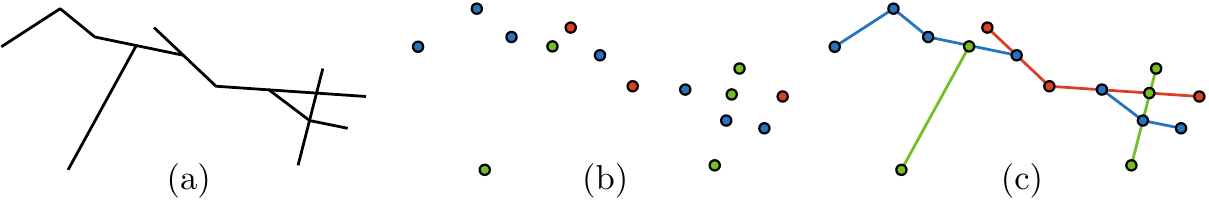}
  \caption{(a) A drawing. 
  (b) A 3-colored point set. (c) The nearest-neighbor graphs.}
  \label{fig:example}
\end{figure}

We are interested in the following problem: suppose 
we are given a \emph{drawing} $\D$, i.e., a set of straight-line segments in the plane such that if two segments intersect, then their intersection is a point and the two segments are not parallel. Under this assumption, by considering a drawing as a set of points $Q$ in the plane, the input segments of $\D$ are the inclusion maximal segments in $Q$.
The \emph{special} points of $\D$ are the endpoints
of the segments in $\D$ and the intersection points between
pairs of segments in $\D$. We denote the set of special points by $P$.
We require that the pairwise distances between the special points of 
$\D$ are all distinct.
Our general task is to find a set of colors $C$ and
a color assignment $\sigma$, such that the union
$\NN$ of the nearest-neighbor graphs for $P$ and $C$
equals $\D$,
interpreted as subsets of the plane.
We call $\NN$ an NN-\emph{decomposition}
of $\D$ with vertex set $P$, where $NN$ stands for \emph{Nearest-Neighbor}
and we call $|C|$ the \emph{color-number} of 
$\NN$: see \cref{fig:example}.

\begin{figure}[tbp]
  \center
  \includegraphics[page=4]{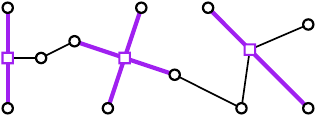}
  \caption{The 
  possible violations that make a drawing non-plane.}
  \label{fig:plane-vs-non-plane-drawings}
\end{figure}

A drawing $\D$ is called \emph{plane} if its segments meet only
at their endpoints, i.e., no segment of $\D$ contains a 
special point in its relative interior; 
see Figure~\ref{fig:plane-vs-non-plane-drawings} for an illustration
where the bold edges contain an endpoint or crossing point (marked with
a square) in their interior which is not allowed by the definition.

Let $\mathcal{C}$ be a connected component in a plane
drawing $\D$, and let $p$ be a special point in $\mathcal{C}$. 
We denote by $a(p)$ the special point  in 
$\mathcal{C} \setminus \{ p\}$ that is closest to $p$ (with distance $d$). 
Let $b(p)$ be the set of special points in $\D$ 
whose distance to $p$ is strictly less than $d$.
By definition, 
$b(p) \subset \mathcal{D} \setminus \mathcal{C}$.  
Let $\mathcal{C}_1$ and $\mathcal{C}_2$ be two distinct 
connected components. We say that $\mathcal{C}_1$ and 
$\mathcal{C}_2$ are \emph{conflicting} if there is a special point 
$p \in \mathcal{C}_1$ such that 
$b(p) \cap \mathcal{C}_2 \neq \emptyset$,
or vice-versa.
We denote by 
$V = \{\mathcal{C}_i\}_{1 \leq i \leq n}$ the connected components of 
$\D$. We define $E$ as the set of pairs $\{\mathcal{C}_i,\mathcal{C}_j\}$ 
where $\mathcal{C}_i$ and $\mathcal{C}_j$ are conflicting. 
We say that the graph $G:=(V,E)$ is the 
\emph{conflict graph} of $\mathcal{D}$.
We call the connected component $\mathcal{C}$ NN-\emph{representable} if 
$\mathcal{C}$ is the nearest-neighbor graph of its
special points. An abstract graph is a \emph{conflict graph} if it is the conflict graph of some plane drawing.

\subparagraph*{Related work.}
The nearest-neighbor graph of a planar point set $P$ is 
well understood~\cite{eppstein97,MM17}. It is a subgraph of the 
relative neighborhood graph of $P$~\cite{jt-rng-92,MM17}, which in turn 
is a subgraph of the Delaunay triangulation. The problem 
of recognizing whether a given abstract graph can be realized as
a nearest-neighbor graph of a planar point set is open 
and we conjecture it to be hard. 
In contrast, testing whether a given embedded graph is a (single) 
nearest-neighbor graph is easy, as it suffices to test if each vertex is indeed conencted to its closest point.

Our problem also has applications in automated content generation 
for puzzle games: van Kapel introduces a version 
of \emph{connect-the-dots} puzzles where the task 
is to connect dots based on colors rather than numbers~\cite{thesis/tim}. 
In this puzzle, points may have multiple colors; 
see \cref{fig:example-puzzle}. 
Van Kapel implemented a heuristic approach for generating 
such puzzles. The heuristic works well for small instances, 
but for larger instances, it 
generates too many colors to be practical~\cite{lkkkks-ctdfp-14}.

\begin{figure}[h]
  \center
  \includegraphics{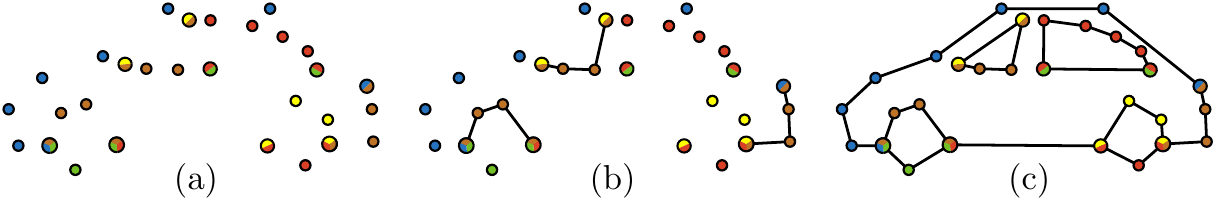}
  \caption{(a) Multi-colored points 
  with 5 colors: blue, red, green, yellow, and orange: (b) 
  the orange nearest-neighbor graph. (c) 
  The union of all nearest-neighbor graphs. Figure 
  taken from~\cite {lkkkks-ctdfp-14}.}
  \label{fig:example-puzzle}
\end{figure}

\subparagraph*{Our Results.}
First, we consider the problem of testing whether a given drawing $\D$ 
can be decomposed into $c$ nearest-neighbor graphs. We show that under the assumption that the drawing is plane, meaning that segments in $\D$ may only meet at their endpoints, this problem is in P for $c \leq 2$, and NP-complete for $c \geq 3$. 
If we allow the segments of $\D$ to cross 
the problem is already NP-complete for $c=2$.

Inspired by our algorithms, we also introduce
the new graph class of \emph{conflict graphs} of drawings.
We show that string graphs are conflict graphs and give additional,
more efficient constructions in terms of size complexity for planar graphs and complete $k$-partite graphs.
On the other hand, subdivisions of non-planar graphs are not conflict graphs.

We show a separator lemma for conflict graphs, which allows us to provide an algorithm for computing a maximum independent set in conflict graphs in subexponential time. Using it as a subroutine, we obtain an $O(\log n)$-approximation algorithm for coloring conflict graphs that runs in subexponential time. This problem is of importance to us because we show that coloring conflict graphs is equivalent to partitioning a plane drawing into nearest-neighbor graphs.

\section{Existence of NN-Decompositions on Special 
Points}

\subsection{The Plane Case}

Let $\D$ be a straight-line drawing. If $s$ is a line segment with $s \subset \D$ such that
$s$ is not a segment of $\D$, we say that $s$ is \emph{covered} by 
$\D$. Recall that the vertex set of the NN-decomposition consists of the special points
in $\D$.
We investigate the question under which circumstances it is possible
to find such a NN-decomposition of $\D$.

\begin{lemma}\label{lem:planeSpecialPoints}
Let $\D$ be a plane drawing. 
Suppose there is a NN-decomposition $\NN$ of
$\mathcal{D}$, and
let $\sigma$ be the underlying coloring
of $\NN$. Then,
for any connected component $\mathcal{C}$ of $\D$, 
the coloring $\sigma$ assigns the same color to all special points
in $\mathcal{C}$.
\end{lemma}

\begin{proof}
Suppose $\D$ has a connected component $\mathcal{C}$ in
which $\sigma$ assigns two distinct colors.
Then, $\mathcal{C}$ has a segment $s = uv$ between
two special points $u$ and $v$ such that
$\sigma(u) \neq \sigma(v)$.
However, the line segment $uv$ must
be covered by $\NN$, and thus, there exists a segment $t$ in $\NN$ 
that contains $u$, $v$, and another special point of $\D$ (since
the segments in $\NN$ are derived from nearest-neighbor relations between
points of the same color). By our assumption that $\D$ is a plane drawing,
the segment $t$ is not in 
$\D$, so $\NN$ is not an NN-decomposition of $\mathcal{D}$, 
a contradiction.
\end{proof}

\begin{theorem}
Let $C$ be a set of colors with
$|C| \leq 2$. 
There is a polynomial-time algorithm for the following task:
given a plane drawing $\D$, is there a NN-decomposition of $\D$ with color set $C$?
\end{theorem}

\begin{proof}
Let $\D$ be a plane drawing.
If there is a decomposition of
$\D$ with color set $C$, then, by
\cref{lem:planeSpecialPoints},
every connected component is colored with a single color of $C$,
i.e., every connected component of $\mathcal{D}$ is 
NN-representable.
The latter necessary condition can be checked in polynomial time,
as we only need to compute the nearest-neighbor graph of the
special vertices in each component. If there is a connected component
where this is not the case,
the algorithm answers that there is no solution. 

Otherwise, we construct the 
conflict graph $G$ of $\D$, and we check 
if $G$ can be colored with $C$. This takes polynomial
time since $|C|\leq 2$  (for $|C|=2$, check whether $G$ is
bipartite, for $|C|=1$, check that $G$ has no edges).
Now, if $G$ is $C$-colorable,
we give all special points in a component $\mathcal{C}$ 
the color assigned to the corresponding vertex in $G$.
Since all connected components are NN-representable, 
this is also a NN-decomposition of $\mathcal{D}$
with $C$.
On the other hand, if $\D$ has a NN-decomposition 
with color set $C$, then $G$ must be $C$-colorable, by
definition of $G$.
\end{proof}

\begin{theorem}\label{thm:NP3colors}
Let $C$ be a set of colors with
$|C| \geq 3$. 
The following task is NP-complete:
given a plane drawing $\D$, is there a NN-decomposition of $\D$ with color set $C$?
\end{theorem}

\begin{proof}
Gr{\"a}f, Stumpf, and Wei{\ss}enfels~\cite{graf1998coloring} 
showed how to reduce $k$-colorability to $k$-colorability of unit disk 
graphs. 
Our proof is inspired by theirs. 
Let $k = |C|$.
We show the NP-hardness of coloring the special points 
of $\D$ with $k \geq 3$ colors  
by means of a reduction from $k$-colorability. 
We make use of four types of gadgets: $k$-wires, $k$-chains, $k$-clones, 
and $k$-crossings. 
They are depicted in \cref{fig:wire,fig:chain,fig:clone,fig:crossingGadget}, 
together with their conflict graphs. 
The symbol consisting of a number $x$ in a circle denotes a 
clique of size $x$. 
A vertex $v$ connected to such a symbol means that there is an 
edge between $v$ and all the vertices of the clique. 
These conflict graphs are exactly the gadgets defined by 
Gr{\"a}f, Stumpf, and Wei{\ss}enfels. 
Note that each connected component in these gadgets is NN-representable. 
The gadgets shown are for $k=5$.

\begin{figure}[tbp]
    \centering
    \includegraphics[page=1,width=0.4\textwidth]{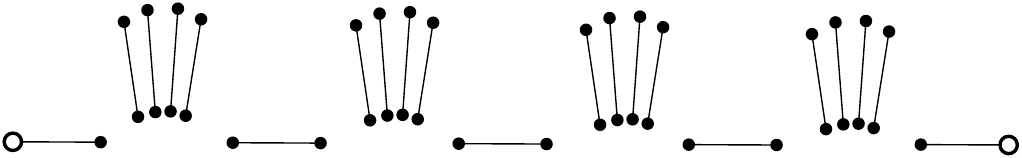}
    \hfill
    \includegraphics[page=2,width=0.4\textwidth]{wire}
    \caption{A $5$-wire of length $5$ and the conflict 
    graph of a $k$-wire of length $5$. The symbol consisting of a number $x$ in a circle denotes a 
clique of size $x$.}
    \label{fig:wire}
\end{figure}

\begin{figure}[tbp]
    \centering
    \includegraphics[page=1,width=0.45\textwidth]{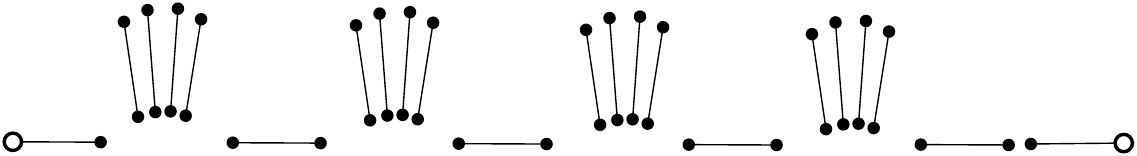}
    \hfill
    \includegraphics[page=2,width=0.45\textwidth]{chain}
    \caption{A $5$-chain of length $5$ and the conflict 
    graph of a $k$-chain of length $5$.}
    \label{fig:chain}
\end{figure}

In \cref{fig:wire,fig:chain,fig:clone}, there are several sets 
of four segments that are very close and nearly vertical. 
For other values of $k$, the gadgets are analogous, 
but with $k-1$ almost vertical segments instead of four.
Similarly, in Figure~\ref{fig:crossingGadget}, there are five sets 
consisting of three close segments. For other values of $k$, 
there are five sets of $k-2$ segments. 
In \cref{fig:wire,fig:chain}, $k$-wires and $k$-chains are drawn 
as if they were on a line, but they may also bend with a right angle.
Note that in \cref{fig:wire,fig:chain,fig:clone,fig:crossingGadget}
some vertices are specially marked with larger empty circles.
These vertices will be called \emph{extreme vertices}.

\begin{figure}[tbp]
    \centering
    \includegraphics[page=1,width=0.45\textwidth]{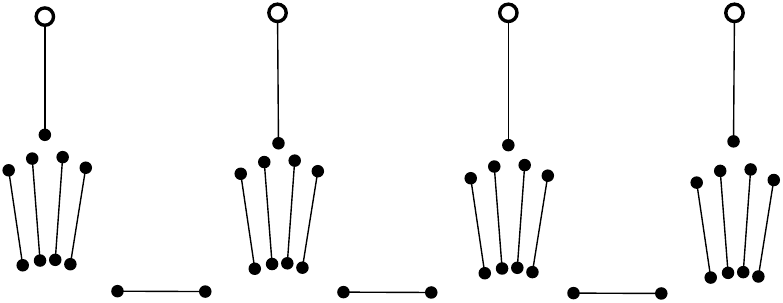}
    \hfill
    \includegraphics[page=2,width=0.45\textwidth]{clone}
    \caption{A $5$-clone of length $4$ and and the 
    conflict graph of a $k$-clone of length $4$.}
    \label{fig:clone}
\end{figure}

\begin{figure}[tbp]
    \centering
    \includegraphics[page=2,width=0.37\textwidth]{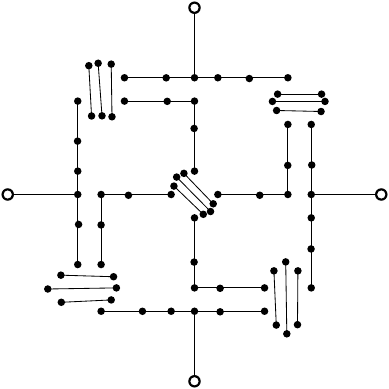}
    \hfill
    \includegraphics[page=2,width=0.37\textwidth]{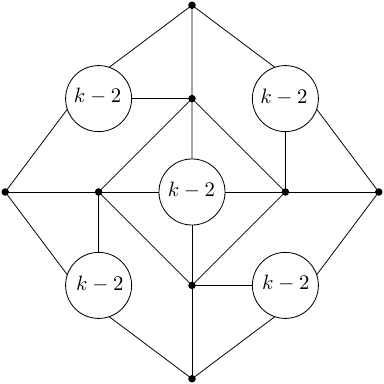}
    \caption{The $5$-crossing gadget and
    the 
    conflict graph of the $k$-crossing gadget.}
    \label{fig:crossingGadget}
\end{figure}

In Figure~\ref{fig:crossingGadget}, there seem to be points 
lying on a segment between two other points. Actually, these 
points are shifted by a sufficiently small $\varepsilon > 0$,
to ensure that the resulting drawing is plane.

The main property of a $k$-wire is that in any coloring with $k$ colors of 
its conflict graph, 
the extreme vertices are assigned the same color. 
In contrast, in a $k$-chain, the extreme vertices are 
assigned different colors. 
In a $k$-clone of length $\ell$, there are $\ell$ extreme vertices. 
In any coloring with colors from $C$, all extreme vertices 
have the same color. Finally, for the $k$-crossing, opposite 
extreme vertices must have the same color;
a pair of consecutive extreme vertices (e.g., top and left 
extreme vertices) may or may not be assigned the same color, as shown in~\cite{johnson1976NPcomplete}.

\begin{figure}[tbp]
  \centering
  \includegraphics[page=1,width=0.2\textwidth]{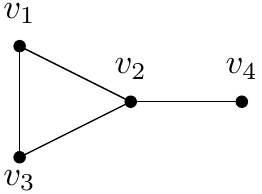}
  \hfill
  \includegraphics[page=2,width=0.7\textwidth]{orthogonal_kColour_graph}
  \caption{A graph with four 
  vertices (left). Converting it to an NN-graph (right).}
  \label{fig:orthogonal-k-color-graph}
\end{figure}

Now we follow the proof of Gr{\"a}f, Stumpf, and Wei{\ss}enfels.
Suppose we are given a graph $G=(V, E)$.
We describe a drawing $\mathcal{D}$ whose conflict graph can be 
colored with color set $C$ if and only if the vertices of $G$ can be 
colored with $C$.
Refer to Figure~\ref{fig:orthogonal-k-color-graph}.
For each vertex $v$ of degree $\delta$ in $G$, we draw a 
$k$-clone of size $\delta$.
The clones are drawn so that they are arranged on a 
horizontal line and such that their upper points have the same $y$-coordinate.
Then, for each edge $\{u,v\} \in E$, we draw 
it on the plane as two vertical segments, each incident to one $k$-clone, 
and one horizontal segment that connects the two upper points of 
the vertical segments.
We do that such that for any pair of edges, their horizontal 
segments have distinct $y$-coordinates.
Then we replace each crossing between a pair of edges by a $k$-crossing.
Finally, let us consider one edge $\{u,v\} \in E$, and let us orient 
it arbitrarily, say toward $v$.
We replace each part of the edge between two $k$-crossings by $k$-wires 
of sufficient length.
If there are no crossings, we replace the edge by a $k$-chain.
Otherwise, the part of the edge between $u$ and the first $k$-crossing 
is replaced by a $k$-wire, and the part between the last 
$k$-crossing and $v$ is replaced by a chain.
As the points of distinct gadgets are sufficiently remote 
(except for pairs of gadgets that are connected on purpose), 
the conflict graph of this drawing is the union of the 
conflict graphs of the individual gadgets.

It is possible to find positions with a polynomial number of bits such
that all pairwise distances are distinct but at the same time the positions
are sufficiently close to the prescribed positions.
This concludes the reduction.

It is straightforward to see that the problem is in NP with the certificate being a coloring of the vertices.
For each point we can easily find its closest point with the same color (we compare squared distances to avoid taking square roots) and add the edges to the resulting graph.
We can then compare the edges with the segments of the original drawing.
\end{proof}

\subsection{The non-plane case}

We show that if drawings are not required to be plane, 
the problem is 
hard for two colors.

\begin{theorem}\label{thm:NPhardness2colors}
Let $C$ be a set of colors with
$|C| = 2$. 
The following problem is NP-complete:
given a drawing $\mathcal{D}$, is there a NN-decomposition of $\mathcal{D}$ with color set $C$?
\end{theorem}

\begin{proof}
\begin{figure}[tbp]
  \centering
  \includegraphics[page=2,scale=0.6]{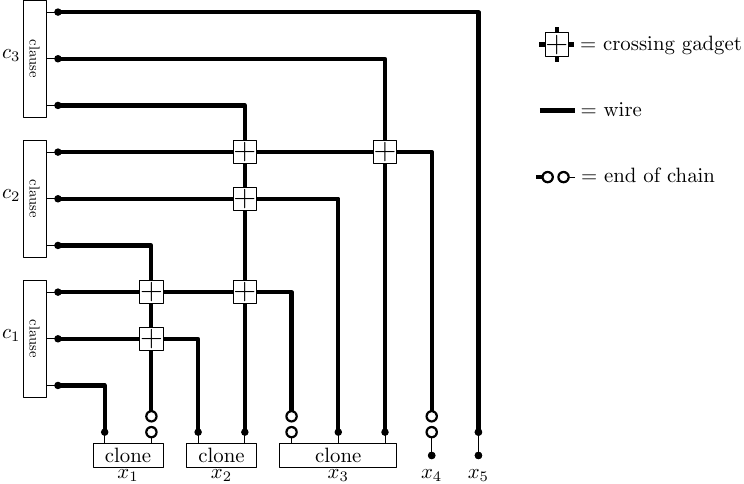}
  \caption{Structure of the conversion of the 
  NAE-3SAT formula with clauses $c_1=(x_1, x_2, \lnot x_3)$,  
  $c_2=(\lnot x_1, x_3, \lnot x_4)$, and  $c_3=(x_2, x_3, x_5)$ 
  into a 2-color $\NN$ graph.}
  \label{fig:two-color-graph}
\end{figure}

We reduce from Not-All-Equal 3SAT (NAE-3SAT), where each clause has three variables and is satisfied if not all variables are equal. 
Let $\Phi$ be an NAE-3SAT formula with variable set $X$ and clause
set $Y$. Let $G_\Phi$ be the associated bipartite graph with
vertex set $X \cup Y$, where two vertices $x$ and $y$ are adjacent 
if and only if $x$ is a variable that appears in clause $y$. 
We draw $G_\Phi$ as follows: clauses are represented 
by vertical segments on the $y$-axis of length $3$. 
Variables of degree $\delta$ are represented as horizontal segments 
on the $x$-axis of length $\delta$. 
Each edge $\{x, y\}$ is drawn as the union of one 
vertical and one horizontal segment. 
The vertical segment is incident to the variable gadget for $x$. 
The horizontal segment is incident to a clause gadget for $y$.
See Figure~\ref{fig:two-color-graph} for an example.

\begin{figure}[tbp]
  \begin{subfigure}{.3\textwidth}
      \centering
      \includegraphics[page=3,width=\textwidth]{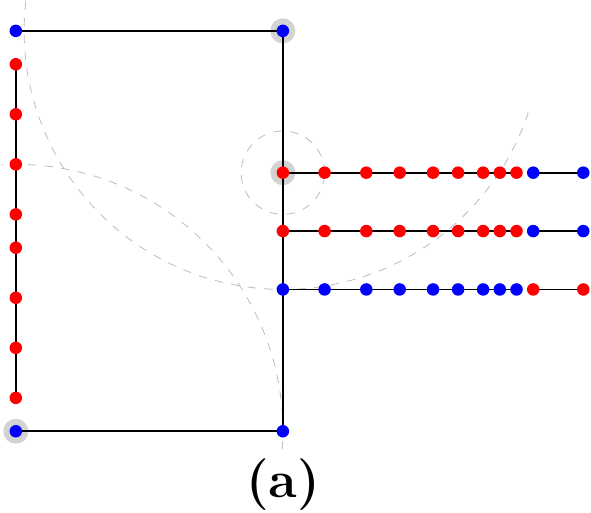}
  \end{subfigure}
  \hfill
  \begin{subfigure}{.3\textwidth}
      \centering
      \includegraphics[page=4,width=\textwidth]{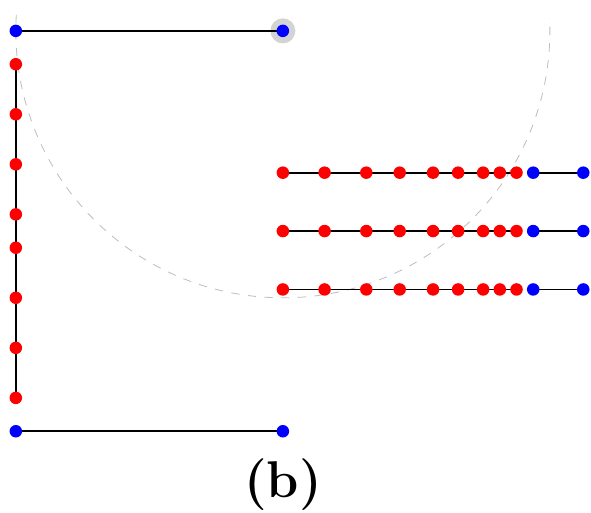}
  \end{subfigure}
  \hfill
  \begin{subfigure}{.3\textwidth}
      \centering
      \includegraphics[page=5,width=\textwidth]{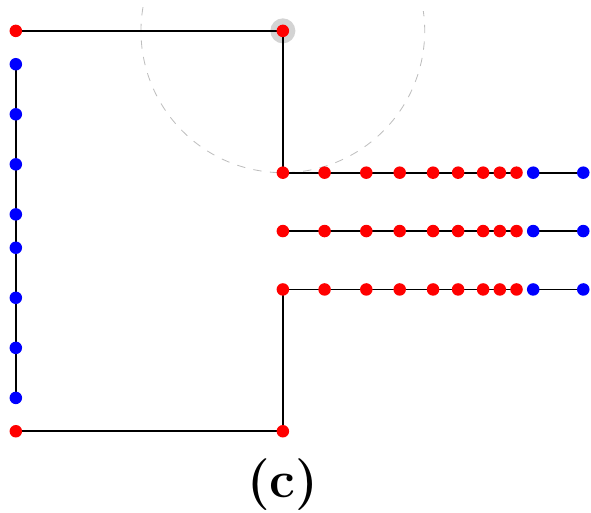}
  \end{subfigure}
      \caption{A clause gadget with (a) a valid assignment, (b)--(c) two invalid assignments. The dashed circles indicate distance to the nearest neighbor.}
      \label{fig:clause_all}
  \end{figure}

\begin{figure}[tbp]
    \centering
    \includegraphics[page=2,scale=0.6]{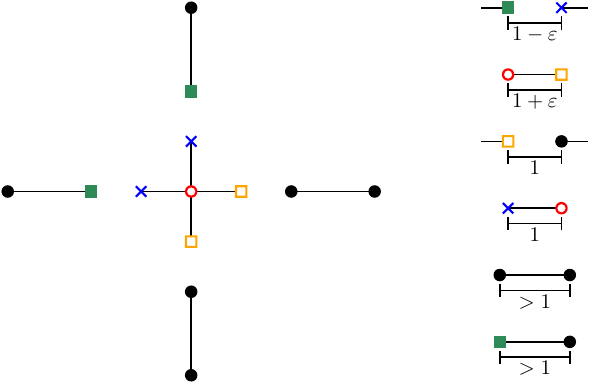}
    \caption{A non-plane crossing gadget.}
    \label{fig:crossingDegenerate}
\end{figure}

We use some gadgets from
the proof of Theorem~\ref{thm:NP3colors}.
We replace each variable by a $2$-clone of length $\delta$.
We replace each clause by the gadget in Figure~\ref{fig:clause_all}a
(see Figure~\ref{fig:clause_all}b-c for assignments where all literals
have the same color). In Figure~\ref{fig:clause_all}a, there seem to be points 
lying on a segment between two other points. Actually, these 
points are shifted by a sufficiently small $\varepsilon > 0$,
to ensure that the resulting drawing is plane.
We replace each crossing by the gadget in Figure~\ref{fig:crossingDegenerate}.
In Figure~\ref{fig:crossingDegenerate}, some points have been colored.
Note that this does not correspond to an assignment of truth
values, but is supposed to provide visual information for the reader.
The distance between a green point and a blue point is 
$1 - \varepsilon$, for a sufficiently small $\varepsilon > 0$.
The distance between a blue point and the red point is $1$.
The distance between the red point and an orange point is $1 + \varepsilon$.
The blue point on the left and the orange point on the right are 
finally shifted by a suitable $\eta > 0$ with $\eta \ll \varepsilon$, 
so that no two points are at the same distance from the red point.
The points in the clause gadget that are on the vertical connected 
component on the left side are arranged so 
that this connected component is NN-representable.
Finally, each part of an edge between two gadgets 
is replaced by a $2$-wire of suitable length.
We have thus obtained a drawing $\mathcal{D}$.

\begin{figure}[tbp]
\hfill
\begin{subfigure}{.3\textwidth}
    \centering
    \includegraphics[page=2,width=0.8\textwidth]{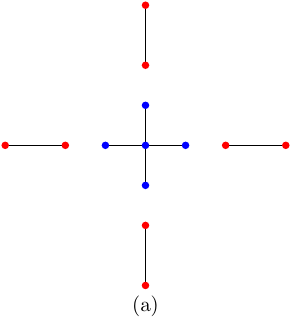}
\end{subfigure}
\hfill
\begin{subfigure}{.3\textwidth}
    \centering
    \includegraphics[page=2,width=0.8\textwidth]{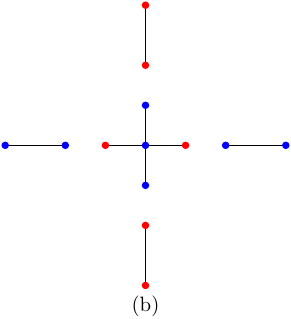}
\end{subfigure}
\hfill
\null
    \caption{Two valid assignments of the non-plane crossing gadget. (a) All extreme segments have the same color. (b) Opposite segments have the same color.}
    \label{fig:crossingDegenerate_assign}
\end{figure}

We claim that $\Phi$ is satisfiable if and only if 
there exists a special-point NN-decomposition of $\mathcal{D}$ 
with two colors. First, notice a clause gadget has a special-point
NN-decomposition if and only if two of the horizontal segments 
on the right side are assigned different colors. Indeed, we show in Figure~\ref{fig:clause_all}b-c that if all literals have the same color, then the corresponding NN-graph is not the one that is required, the one shown in Figure~\ref{fig:clause_all}a. It remains to show that if not all literals have the same color, then we obtain the correct NN-graph. By symmetry, if the top and bottom literals do not have the same color, then we are in the situation of Figure~\ref{fig:clause_all}a. If the top and bottom literals have the same color, say red, then we keep the color of the remaining points of the clause gadget as in Figure~\ref{fig:clause_all}a. Let us denote by $p$ the point of the clause gadget incident to the middle literal. By assumption, $p$ is blue. Let us denote by $q$ the top right vertex of the clause gadget, which is also colored in blue. Therefore the closest neighbohr of $q$ which is also blue is $p$. Likewise, in the NN-graph, the closest neighbohr colored in blue of the bottom right vertex is $p$. This shows that in this situation, the NN-graph is the same as in Figure~\ref{fig:clause_all}a.

In the non-plane crossing gadget opposite segments are assigned the same 
color. All of them may be assigned the same color, as in 
Figure~\ref{fig:crossingDegenerate_assign}a, or consecutive segments 
might be assigned different colors, as in 
Figure~\ref{fig:crossingDegenerate_assign}b. Therefore, by 
associating the colors of $C$ with truth values, $\mathcal{D}$
has a special-point NN-decomposition if and only if 
$\Phi$ is satisfiable.

That the problem is in NP can be seen the same way as in the proof for \cref{thm:NP3colors}.
\end{proof}

\section{Conflict Graphs and Related Graph Classes}

We show in Theorem~\ref{thm:NP3colors} that $k$-coloring 
of conflict graphs is NP-complete, for any fixed $k \geq 3$.
To put this result into context, 
we will show that there exist graphs that are not conflict graphs.
Moreover, we will prove the inclusion of some well-known 
graph classes in the class of conflict graphs. The aim is to characterise the class of conflict graphs, as it gives some information about what kind of running time we can expect for the vertex coloring algorithms on conflict graphs.

Let $G$ be a graph, and let us denote by $G'$ the graph obtained 
from $G$ by subdividing each edge of $G$ once 
(i.e., for each edge $e=\{u,v\}$ in $G$, we add a 
vertex in $G'$ whose neighbors are exactly $u$ and $v$). 
In this section we will show that $G'$ is a conflict graph if and only if $G$ is planar.
Sinden~\cite{sinden1966Topology} showed the same statement for $G'$ being a string graph,
i.e., an intersection graph of continuous curves in the plane.

First, let us recall Sinden's proof for string graphs. 
Let $G$ be a graph, and let $G'$ be obtained from $G$ as
described above.
Assume that $G'$ is a string graph, and consider a representation $R$ of 
$G'$ as a string graph. Contract to a point each curve in $R$ that corresponds 
to a vertex in $G$ and extend the curves corresponding to the edges 
in $G$. In the process, one can maintain the property that the 
curves corresponding to the edges in $G$  intersect only two other 
curves, which are now reduced to single points. Observe that the resulting 
drawing is a plane embedding of $G$. We adapt this method
for conflict graphs.

Let $G$ be a graph, and let $G'$ be the edge-subdivision of $G$. 
Suppose that $G'$ is a conflict graph, and let $R$ be a representation of 
$G'$ as a conflict graph. We still would like to contract
each connected component of $R$ that corresponds to a vertex in $G$ to
a single point. However, now it is not clear that we can extend 
the connected components corresponding to edges in $G$ such that 
they only intersect two other connected components (now reduced to points). 
This is illustrated in \cref{fig:planarSubdivision}. 
The connected components in blue correspond to the vertices in $G$, 
those in red correspond to the edges in $G$. The dashed segments 
show a conflict between two connected components. Inside the green square, 
we have a connected component blocking another one. We therefore 
want to reroute the connected component so that there is no intersection. 
To show how we do it, we first need the following lemmas.

\begin{figure}[h]
\centering
\includegraphics[scale=0.48]{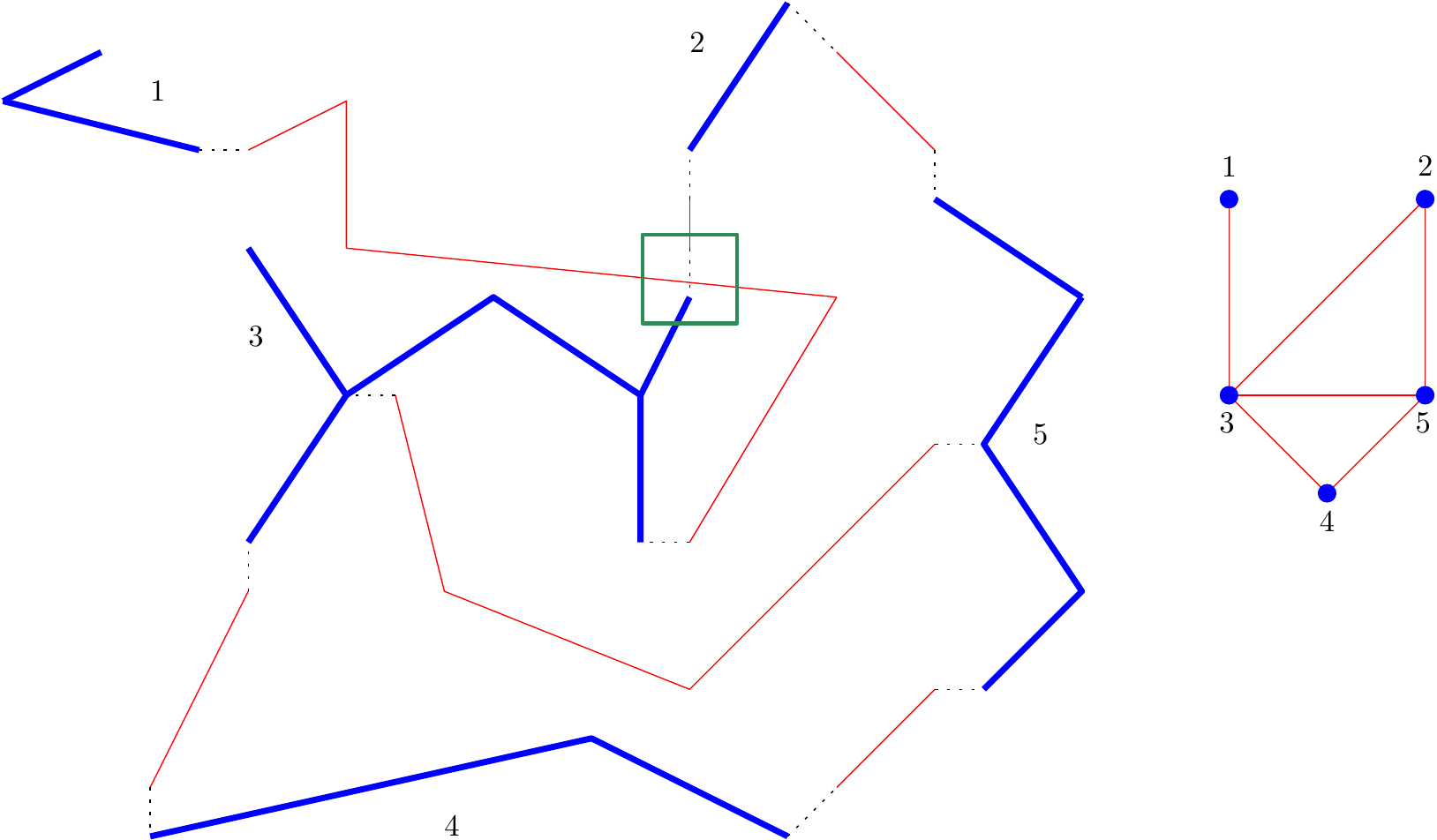}
\caption{A representation of the subdivision of $G$ ($G$ is shown on the right side) as a conflict graph. The connected components in blue correspond to vertices in $G$, the ones in red correspond to edges in $G$. The dashed segments show a conflict between two connected components. In the green zone, a conflict overlaps with a connected component.}
\label{fig:planarSubdivision}
\end{figure}

Let $\D$ be a plane drawing. Let $G_\D$ be its conflict graph. 
Let $\D_1, \D_2$ be a pair of connected components that are conflicting. 
Let us consider two points $p$ and $q$ that certify this conflict 
(for instance, assume without loss of generality $p \in \D_1$, $q \in \D_2$,  
$q \in b(p)$). We denote by $s$ the segment with endpoints $p$ and $q$.

\begin{figure}[h]
\centering
\includegraphics[scale=0.8]{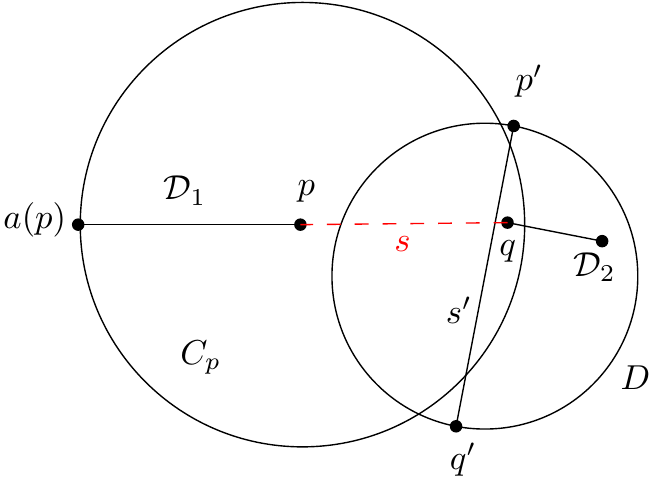}
\caption{Illustration of Lemma~\ref{lem:ConCompIntersect}. The disk $D$ contains $q$.}
\label{fig:ConCompIntersect}
\end{figure}

\begin{lemma}\label{lem:ConCompIntersect}
If a connected component intersects the line segment $s$, then this connected component is conflicting with $\D_1$ or $\D_2$.
\end{lemma}

\begin{proof}
    The proof is illustrated in Figure~\ref{fig:ConCompIntersect}. Let us denote by $s'$ a segment in a component $\D_3$ which intersects $s$. Let us assume that $\D_3$ is not conflicting with $\D_1$, and let us show that it is conflicting with $\D_2$. Let us consider the circle $C_p$ centered at $p$ going through $a(p)$, the closest point to $p$ among the ones that are in the same connected component. By assumption, the endpoints of $s'$, denoted by $p'$ and $q'$, are not inside $C_p$. In contrast, $q$ is inside $C_p$. We consider the disk $D$ with diameter $s'$. To show that $\D_3$ is conflicting with $\D_2$, it is sufficient to show that $q$ is contained in $D$, for then we have $q \in b(p')$ or $q \in b(q')$. We know that $s'$ intersects $C_{p}$ twice. By assumption, $p$ is not contained in $D$. As two circles can intersect at most twice, $D$ contains all of $C_p$ on at least one side of $s'$. Therefore, $D$ contains~$q$.
\end{proof}

Let us consider another pair of conflicting components, denoted by  $\D_3$ and $\D_4$. Let $u \in \D_3$ and $v \in \D_4$ be two points that certify this conflict. Let us assume without loss of generality $v \in b(u)$. We denote by $s'$ the segment with endpoints $u$ and $v$.

\begin{lemma}\label{lem:TwoSegmentsIntersect}
Assume that $u$ and $v$ are not in $b(p)$. If the segments $s$ and $s'$ intersect, then $p$ or $q$ is in $b(u)$. In particular, at least one of $\D_3$ and $\D_4$ is conflicting with $\D_1$ or $\D_2$.
\end{lemma}

\begin{proof}
The situation is similar to the one of Lemma~\ref{lem:ConCompIntersect}. By assumption, the segment $s'$ intersects twice the circle $C_p$ centered at $p$ with radius $a(p)$. Therefore the proof of Lemma~\ref{lem:ConCompIntersect} can be applied in this situation, too.
\end{proof}

We are now ready to prove the theorem. As we are considering conflict graphs, there might be obstacles when we try to follow Sinden's proof that were not there with string graphs. We use the two lemmas above to reroute these obstacles.

\begin{theorem} \label {thm:subdivision}
Let $G$ be a graph and let $G'$ denote the subdivision of $G$. If $G'$ is a conflict graph then $G$ is planar.
\end{theorem}

\begin{proof}
Let us assume we have a representation of $G'$ as a conflict graph. We denote by $V$ the vertices in $G'$ corresponding to vertices in $G$ and by $E$ the vertices in $G'$ corresponding to edges in $G$. For each vertex $v\in V$, we pick an arbitrary point $p_v$ on the connected component that represents $v$ in the conflict representation. We are going to reduce all connected components $v$ to their corresponding point $p_v$. Let $e\in E$ be the vertex in $G'$ corresponding to the edge $\{w,x\}$ in $G$. We want to extend  $e$ to a curve that contains $v_x$ and $v_w$ at each endpoint. We want to do that for all vertices in $E$, such that no two curves intersect, except maybe at endpoints. Therefore, we would obtain a plane representation of $G$.

Let us first consider all connected components in the representation as a conflict graph, before reducing some of them to points, and extending the rest to curves. For each pair of connected components $( \D_1, \D_2)$ that are conflicting, we find two points $p$ and $q$ that certify it, meaning that $p$ is in $b(q)$ or vice versa. We now draw the segment with endpoint $p$ and $q$, for each such pair of conflicting components. We denote by $S$ the set of segments we have obtained. Let us consider a connected component $\cal D$ corresponding to a vertex in $E$, that intersects with some segments in $S$. For one of these segments it intersects, say $s$, we name its endpoints $p$ and $q$. Without loss of generality, we assume $p \in\D_1$ and $q \in \D_2$. By Lemma~\ref{lem:ConCompIntersect}, $\cal D$  is conflicting with $\D_1$ or $\D_2$. By assumption, $\D_1$ and $\D_2$ are conflicting. Therefore one of them, say $\D_1$, corresponds to a vertex in $V$, and the other corresponds to a vertex in $E$. By definition of $G'$, no two vertices in $E$ are connected by an edge. This implies that $\cal D$ is conflicting with $\D_1$. This shows that when we want to extend a connected component $e$ to a curve that contains $v_x$ and $v_w$ at each endpoint (see the notation above), we might be blocked by other curves, but these curves have to also contain $v_x$ or $v_w$ at an endpoint. This is the situation depicted in Figure~\ref{fig:planarSubdivision}. Here we simply have to reroute the edge going from vertex $1$ to $3$. 

One issue that might still occur is that when trying to extend a connected component into a curve, we are not blocked by another connected component, but by the extension into a curve of a connected component. Namely, how do we do the rerouting when two segments in $S$ intersect? Let $s$ and $s'$ be those two segments. By construction, $s$ is a segment between two connected components $\D_1$ and $\D_2$. Without loss of generality, we can assume that $\D_1$ corresponds to a vertex in $V$ and $\D_2$ to a vertex in $E$. Likewise, $s'$ is a segment between two connected components $\D_3$ and $\D_4$ , with $\D_3$ corresponding to a vertex in $V$ and $\D_4$ to a vertex in $E$. Now we use Lemma~\ref{lem:TwoSegmentsIntersect}, which states that one connected component from each pair are conflicting. By construction, this pair is $( \D_1, \D_4)$ or $( \D_2, \D_3)$. This shows that we are in the same situation as in the paragraph above. Thus it is possible to reduce each connected component corresponding to a vertex in $V$ to a point, and then extend the connected components from $E$ into curves. By doing that, we obtain a plane drawing of~$G$.
\end{proof}

Theorem~\ref {thm:subdivision} and its analogy to the proof for string graphs in~\cite{sinden1966Topology} may suggest that the class of conflict graph is equal to the class of string graphs.
We show that, indeed, the string graphs are contained in the conflict graphs, but conjecture that the opposite is not true.
  
  \begin {theorem}
  \label {thm:stringgraphsareconflictgraphs}
  All string graphs are conflict graphs.
  \end {theorem}

  \begin {proof}
  Let $G$ be a string graph.
  We start by embedding $G$ as a set of $n$ strings (curves) in the plane. We may assume the strings are non-self-intersecting, but they could intersect other strings multiple times.
  For ease of exposition, we further assume that all curves are orthogonal polylines aligned to a unit grid.
  The proof steps are illustrated in Figure~\ref{fig:stringproof}.
  
  Fix a resolution $r_0 = \frac12$ such that if we place points on each string at distance $r_0$ from each other, then these points will always be closer to each other than to points on other strings, except near crossings.
  Pick an arbitrary string $s_0$. Consider the \enquote{tunnel} of width $r_0$ around $s_0$; by construction the other strings cross this tube in consecutive proper crossings. Consider the ordered list $c_1, c_2, ..., c_k$ of crossings of $s_0$ with other strings (note that $k$ could be independent of $n$).
  
  \begin{figure}[tp]
    \centering
    \includegraphics{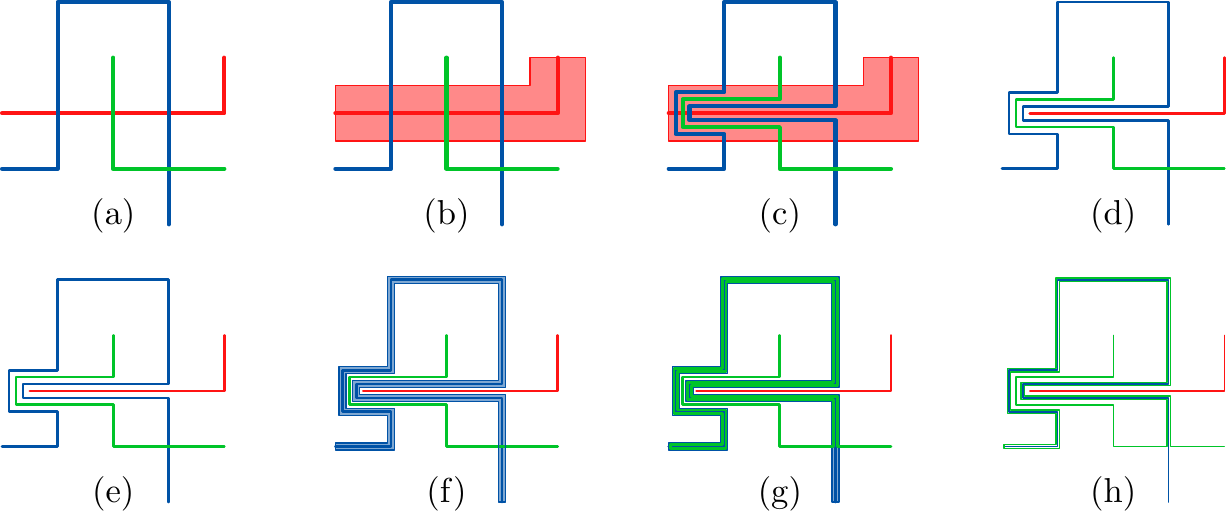}
    \caption{Illustration of the proof of Theorem~\ref {thm:stringgraphsareconflictgraphs}.
    (a) A string representation with $3$ strings.
    (b) The tunnel around the red string.
    (c) Rerouting the blue and green strings inside the red tunnel. Note that the blue string crosses the red string twice.
    (d) The final set of strings after one iteration of the algorithm. The red string no longer intersects any other strings.
    (e-h) The second (and last) iteration of the algorithm.
    }
    \label{fig:stringproof}
  \end{figure}
    
  Set $r_1 = r_0 / 10k$. We increase the resolution of all strings except $s_0$ to $r_1$. Now we reroute all other strings inside the tube of $s_0$ so that they keep distance $r_1$ from each other and from the tube boundary, but such that the crossings with $s_0$ are close to the start point of $s_0$; specifically, crossing $c_i$ should be at distance $2ir_1$ from the start point of $s_0$. Note that this rerouting is always possible.
  
  Now, we shorten $s_0$ by deleting the first $r_0$ length of it. After this, $s_0$ does not intersect any other strings, but if we keep the resolution of $s_0$ at $r_0$, it will have a conflict with exactly those strings that it originally intersected.
  Since the resolutions of the remaining strings were increased, they do not have any conflicts except near crossings with each other.

  We recursively apply this strategy: pick an arbitrary remaining string $s_i$, keep its resolution at $r_i$, determine its crossings and a finer resolution $r_{i+1}$, reroute all strings, and shorten $s_i$. In the end, we will have $n$ strings $s_0,\ldots,s_{n-1}$ at increasingly fine resolutions $r_0,\ldots,r_{n-1}$ whose conflict graph is exactly the original string graph.
  \end {proof}

  Note that the resolution (and hence the size of the components) has a superexponential growth.
  Hence, this does not give us a polynomial-time reduction, and thus we cannot conclude that coloring embedded conflict graphs is NP-hard from the fact that coloring embedded string graphs is NP-hard.

  However, the following lemmas show that both planar graphs as well as complete $k$-partite graphs are conflict graphs and the proof gives a polynomial-time reduction (for a fixed $k$ in the latter case). Note that Lemma~\ref {lem:planarconflict} gives us an alternative proof that coloring embedded conflict graphs is NP-hard.

 \begin{lemma}\label {lem:planarconflict}
  Every planar graph with $n$ vertices is a conflict graph of a set of components of complexity polynomial in $n$.
   \end{lemma}

  \begin{proof}
  Given a plane straight-line drawing $D$ of a planar graph $G$ (obtained by Fáry's theorem \cite{fary1948straight}) such that no two points are on a common vertical or horizontal line.
  We construct a drawing $D'$ such that its conflict graph is $G$.
  The rough idea is to replace every vertex by the gadget in Figure~\ref{fig:vertex_replacement}, and replace each edge by two edges which enforce a conflict.
   
  \begin{figure}
  \centering
  \includegraphics[width=0.3\textwidth]{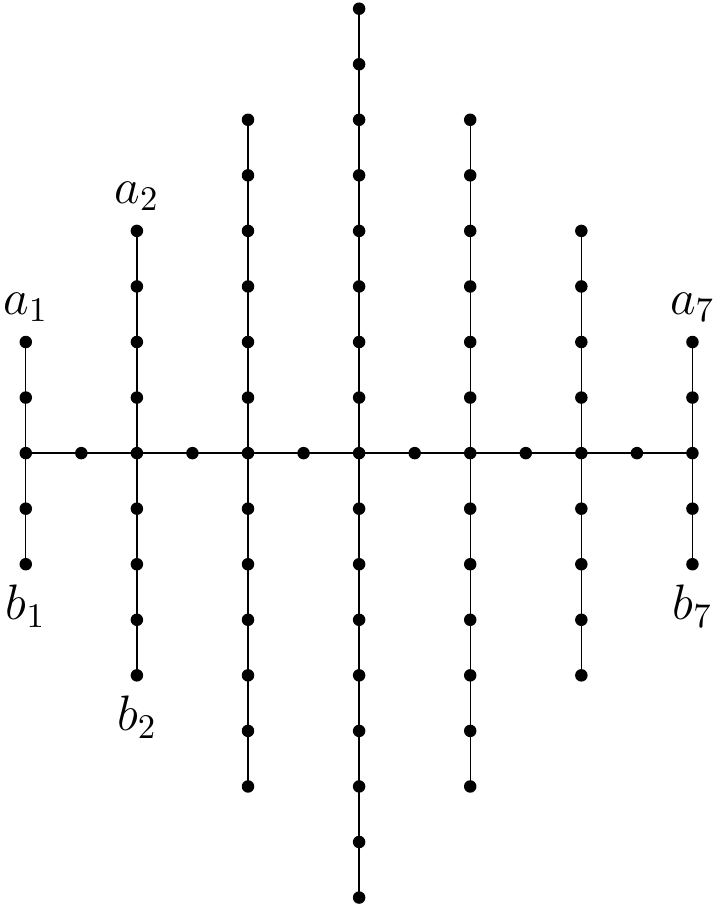}
  \caption{Every vertex is replaced by this gadget. Here for $|V|=4$.}
  \label{fig:vertex_replacement}
  \end{figure}

  The gadget's drawing $H$ (\cref{fig:vertex_replacement}) depends on the number of vertices $n$ of $G$.
  We construct a horizontal line with $4n-3$ points with small distance $\ell$.
  To every odd vertex we add vertical lines, starting with 4 additional vertices on the left and right end, increasing by 4 vertices every step towards the center (with same distance $\ell$).
  The topmost and bottommost vertex of these lines are denoted by $a_k$ and $b_k$ respectively, from left to right.

  We replace every vertex $v_i$ in $D$ with a copy of $H$, called $V_i$.
  Let $a_{i,j}$ and $b_{i,j}$ denote $a_j$ and $b_j$ of $V_i$ respectively.
  Consider an edge $\{v_i,v_j\}$ of $D$.
  Without loss of generality assume $v_j$ is above and to the left of $v_i$ (the other cases are symmetric).
  We consider all vertices in this quadrant of $v_i$ in clockwise order and assume that $v_j$ is the $x$-th such vertex.
  Accordingly, assume $v_i$ is the $y$-th vertex in clockwise order below and to the right of $v_j$.

  We draw the edge $(a_{i,x}, b_{j,2n-y})$ and replace it by a path $P$ of segments of length $\ell$ from $b_{j,2n-y}$ to a point $p$ with $0<d(p, a_{i,x})<\ell$.
  The points $p$ and $a_{i,x}$ enforce the conflict between $V_i$ and $V_j$ and hence the edge $\{v_i,v_j\}$ in $G$.

  The conflict graph of $D'$ obviously contains all the edges of $G$.
  It remains to show, that no additional edges are introduced by showing that any introduced path $P$ only conflicts with the corresponding component $V_i$.
  Assume that the distance between every edge $e$ and every vertex not incident to $e$ is at least $d$.
  If $\ell$ is chosen small enough then the distance of $P$ to every gadget $V_k$ except $V_i$ and $V_j$ is larger than $\ell$.

  Assume $Q$ is a path that conflicts with $P$.
  If $Q$ does neither start at $V_i$ nor at $V_j$, then $Q$ cannot conflict with $P$ since the distance between two independent edges is at least $d$.
  If $Q$ starts at $V_j$ we have two cases.
  If $Q$ is below $P$, then the slope of $Q$ is less than the slope of $P$.
  If $Q$ is above $P$, then the slope of $Q$ is greater than the slope of $P$.
  In both cases the smallest distance is between the starting points of $P$ and $Q$ which is at least $2\ell$ by construction.
  Hence, there does not exist such a further conflict.  
  \end{proof}

   \begin{lemma}
    Every complete $k$-partite graph with $n$ vertices is a conflict graph of a set of components of complexity polynomial in $n$.
   \end{lemma}
  \begin{proof}
    Given a complete $k$-partite graph where each set with index $i$ has $n_i$ vertices.
    We will construct a graph with a drawing $D$ such that $G$ is the conflict graph of $D$, i.e., each set $i$ will consist of $n_i$ components which are not in conflict with each other but with all other components from the other sets.
  
    For each group $i$ the drawing needs two properties: the components' resolution $r_i$ and the components' minimum distance to each other $d_i$.
    The idea of the construction is that $d_i>r_i$ for all $i$ but $r_j>d_i$ for all $j>i$ so that components in the same set do not conflict but all components in a set $j$ conflict with all other components $i<j$ (and thus also with all $i>j$).
  
    \begin{figure}[tp]
      \centering
      \includegraphics[page=1,width=0.3\textwidth]{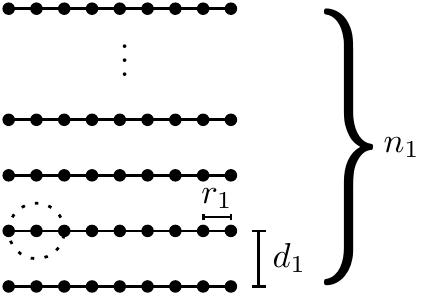}
      \hfill
      \includegraphics[page=2,width=0.48\textwidth]{k-partite}
      \caption{Left: The drawing for the first set. Right: The drawing for the third set with rerouted first set.}\label{fig:k-partite-drawing}
    \end{figure}
  
    We start with $n_1$ parallel components where $r_1=1$ and $d_1=2r_1=2$, see \cref{fig:k-partite-drawing} (left).
    For all following sets $1<i\leq k$ we set $r_i=\frac{3}{2}n_{i-1}d_{i-1}=3n_{i-1}r_{i-1}$ and $d_i=2r_i=3n_{i-1}d_{i-1}=6n_{i-1}r_{i-1}$.
    We place the components from the $i$th set orthogonal to the $(i-1)$th, starting at the opposite corner of a square of size $n_{i-1}d_{i-1}\times n_{i-1}d_{i-1}$, as seen in \cref{fig:k-partite-drawing} (right).
    We finally route the $(i-1)$th set between the last two components of the $i$th set with distance $d_{i-1}$ to the last component.
    This requires that $r_i$ and $d_i$ are multiples of $r_{i-1}$ and that $d_i$ is a multiple of $d_{i-1}$ which in ensured by our choices for $r_i$ and $d_i$.
  
    \begin{figure}[tp]
      \centering
      \includegraphics[page=3,scale=0.2,angle=90]{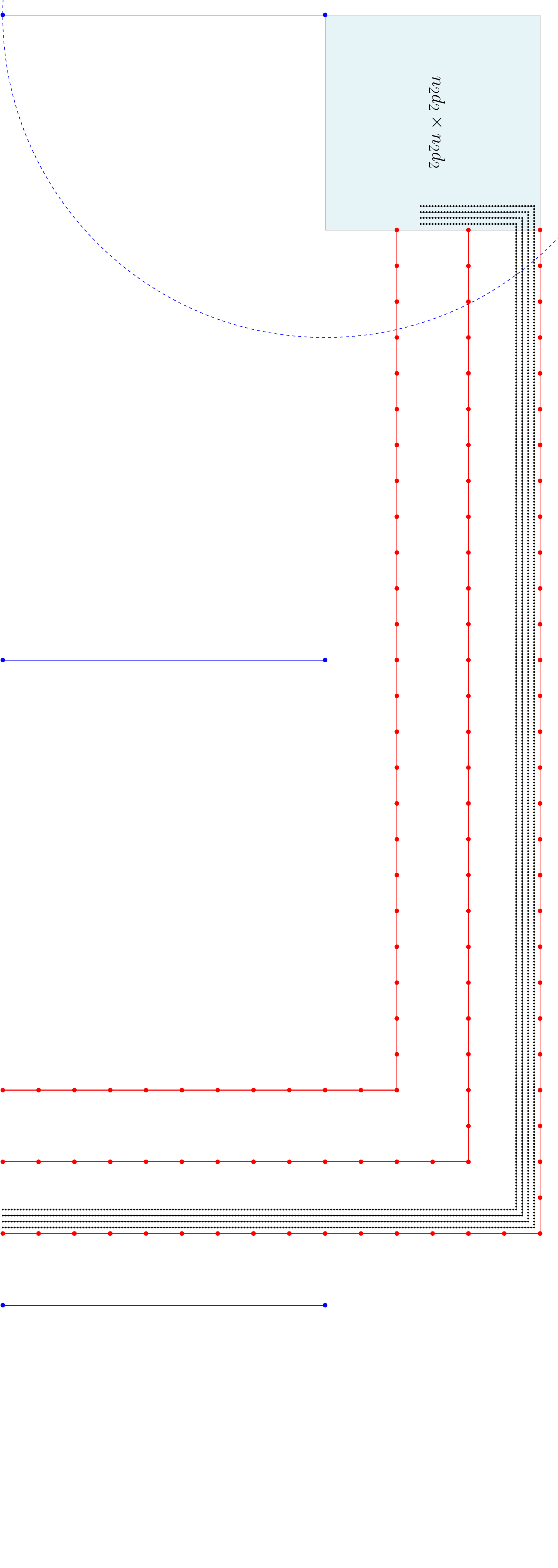}
      \caption{A conflict graph representation of the complete $3$-partite graph $K_{4,3,3}$.}\label{fig:k-partite-drawing-2}
    \end{figure} 

    A depiction of the resulting drawing for $n_1=4$, $n_2=3$, and $n_3=3$ can be found in \cref{fig:k-partite-drawing-2}.
    From the construction it is immediate that a component is not in conflict with a component of the same group but with all other components.
    Hence the initial graph $G$ is a conflict graph of our constructed drawing $D$.
  \end{proof}

\section{Separators of conflict graphs and chromatic number}

Let us recall the motivation for our problem. Starting from a drawing $\D$, we want to color the vertices of that drawing such that the nearest-neighbor graph on those colored vertices is the drawing $\D$. If the drawing is plane, we have shown that this problem boils down to coloring a conflict graph, where one vertex corresponds to a connected component of the drawing. We have shown in Theorem~\ref{thm:NP3colors} that this problem is NP-hard, even for $3$-coloring plane drawings. In this section, we show the following theorem:

\begin{theorem}
There exist an exact algorithm for maximum independent set, and an $O(\log n)$-approximation algorithm for vertex coloring in conflict graphs with $n$ vertices, running in $2^{n^{4/5}\polylog n}$-time.
\end{theorem}

The exact algorithm for maximum independent set is used as a subroutine to obtain the $O(\log n)$-approximation algorithm for coloring conflict graphs in subexponential time, where $n$ denotes the number of vertices in the conflict graph. Indeed, coloring vertices can be seen as a covering problem, where there is a hyperedge for a set of vertices if and only if those are independent. As we have an exact algorithm for maximum independent set, we can use the greedy algorithm for covering to obtain the $O(\log n)$ approximation. 

In~\cite{fox2011computing}, Fox and Pach present an algorithm running in $2^{n^{4/5}\polylog n}$-time for maximum independent set in string graphs. The input is an abstract graph, and it outputs a maximum independent set or a certificate that the input graph is not a string graph. As they observe, the only property they use is the following separator lemma: Every string graph with $m$ edges and maximum degree $\Delta$ contains a separator of order at most $c\Delta m^{1/2} \log m$~\cite{fox2011computing}. A separator in a graph $G=(V,E)$ with $n$ vertices is a subset $V_0 \subset V$ such that there is partition of $V$ into three sets $V=V_0 \cup V_1 \cup V_2$, with $|V_1|\leq 2n/3$ and  $|V_2|\leq 2n/3$, such that there is no edge between a vertex in $V_1$ and a vertex in $V_2$. We show that this lemma also holds for conflict graphs, which immediately implies that the algorithm by Fox and Pach also applies to our setting. The lemma was actually proven in another paper by the same authors, denoted there as Theorem 2.5~\cite{fox2010separator}. The proof uses several lemmas, but the assumption that the considered graph is a string graph appears only once. As defined in~\cite{fox2010separator}, the \emph{pair-crossing number} $\textit{pcr}(G)$ of a graph $G$ is the minimum number of pairs of edges that intersect in a drawing of $G$. Fox and Pach showed that if $G$ is a string graph, then $\textit{pcr}(G)$ is at most the number of paths of length $2$ or $3$ in $G$, where the length of a path is the number of its edges. Following the proof of their Theorem 2.5~\cite{fox2010separator}, it is sufficient for us to show the following:

\begin{lemma}\label{lem:conflictPCR}
 If $G$ is a conflict graph, then $\textit{pcr}(G)$ is at most the number of paths of length $2$ or $3$ in $G$.
\end{lemma} 

\begin{proof}
We use a similar notation to~\cite{fox2010separator}. Let us consider a representation of a conflict graph. For each pair of conflicting connected components $\D_i$ and $\D_j$, we consider two points $p\in \D_i$ and $q\in \D_j$, such that $q\in b(p)$. We now consider the drawing $\D$ consisting of the union of the connected components and the line segments with endpoints $p,q$, for each pair of conflicting connected components. If three or more line segments intersect at the same point, we shift slightly the relative interior of one so that this is not the case anymore. It is not an issue that those are not segments anymore, we only want them to be Jordan curves. For simplicity, we keep referring to them as line segments with endpoints $p,q$. For each connected component $\D_i$, we consider an arbitrary point $p_i$ on $\D_i$. Let $x=\{\D_i,\D_j\}$ denote a pair of conflicting connected components. We denote by $\alpha(x)$ a curve that starts at $p_i$, goes along $\D_i$ until it reaches the line segment between $\D_i$ and $\D_j$ that we have added to the drawing, follows this line segment until it reaches $\D_j$, and finally goes along $\D_j$ until it reaches $p_j$. Observe that this gives us a drawing of $G$ in the plane. Suppose that two edges $\alpha(x)$ and $\alpha(y)$ in this drawing intersect. We claim that they determine a unique path of length $2$ or $3$ in $G$.

First if $x$ and $y$ share a connected component $\D_i$, then they determine a path of length $2$ where $\D_i$ is the middle vertex. Now let us assume without loss of generality that $x=\{\D_1,\D_2\}$ and $y=\{\D_3,\D_4\}$. If the curves $\alpha(x)$ and $\alpha(y)$ intersect inside one of the connected components, say $\D_3$, then we can apply Lemma~\ref{lem:ConCompIntersect} to infer that $\D_3$ is in conflict with one of the components in $x$, say $D_1$.
    If $\alpha(x)$ and $\alpha(y)$ intersect outside of a connected component (i.e., the line segment between $D_1$ and $D_2$ intersects the line segment between $D_3$ and $D_4$) we can similarly apply Lemma~\ref{lem:TwoSegmentsIntersect} to show that one component from $y$, say $\D_3$, is in conflict with one component in $x$, say $D_1$.
In both cases the vertices $\D_2$, $\D_1$, $\D_3$, $\D_4$ form a path of length $3$.
We have shown that a pair of edges that intersect in the drawing determine a unique path of length $2$ or $3$ in $G$.
\end{proof}

Lemma~\ref{lem:conflictPCR} is sufficient to show that conflict graphs with $m$ edges and maximum degree $\Delta$ contain a separator of order at most $c\Delta m^{1/2} \log m$, as shown by Fox and Pach~\cite{fox2010separator}. Following their notation, the \emph{bisection width} $b(G)$ of a graph is the least integer such that there is a partition $V=V_1 \cup V_2$ with $|V_1|$, $|V_2|\leq 2|V|/3$ and the number of edges between $V_1$ and $V_2$ is $b(G)$. As shown by Kolman and Matou{\v{s}}ek~\cite{kolman2004crossing}, we have for every graph $G$ on $n$ vertices $b(G)\leq c \log(n) (\sqrt{\textit{pcr}(G)}+\sqrt{\textit{ssqd}(G)})$ where $c$ is a constant and $\textit{ssqd}(G)$ is twice the number of paths of length $1$ or $2$ in $G$. By denoting by $p$ the number of paths of length at most $3$ in $G$, we derive from Lemma~\ref{lem:conflictPCR} that if $G$ is a conflict graph, then $b(G)=O(p^{1/2}\log n)$. A simple argument proven by Fox and Pach states that $p$ is at most $m\Delta^2$ for a graph with $m$ edges and maximum degree $\Delta$~\cite{fox2010separator}, which concludes the proof.

  \section{Conclusion and open problems}
In this work we studied the decomposition of a drawing into nearest-neighbor graphs.
First, we studied the decision problem, whether for a given natural number $k\geq 2$ it is possible to decompose a drawing into $k$ nearest-neighbor graphs.
If we allow that segments of the drawing cross the problem is NP-complete.
If we assume that the segments only meet at endpoints, it is NP-complete for $k\geq 3$ and polynomial-time solvable for $k=2$. We provided an $O(\log n)$-approximation algorithm running in subexponential time for coloring plane drawings with a minimum number of colors, which we showed to be equivalent to partitioning a plane drawing into a minimum number of nearest-neighbor graph. It would be interesting to find better approximation algorithms, with respect to the approximation ratio or the running time. Also, it would be interesting to study other variants of this problem; specifically, where points can have multiple colors.

We introduced so called conflict graphs and
showed that not every graph is a conflict graph, but every string graph is a conflict graph.
It is an open problem, whether there is a conflict graph, which is not a string graph.
Further it would be interesting to know relations between other graph classes and conflict graphs.

\bibliography{literature}

\end{document}